\documentclass{article}

\usepackage[latin1]{inputenc}
\usepackage[english]{babel} 
\usepackage{amsmath,amsfonts,amssymb,amsthm}

\newtheorem{thm}{Theorem}

\newtheorem{lemma}[thm]{Lemma}
\newtheorem{prop}[thm]{Proposition}

\theoremstyle{definition}
\newtheorem{defn}[thm]{Definition}
\newtheorem{remark}[thm]{Remark}

\newcommand{\abs}[1]{\left\vert {#1} \right\vert}
\newcommand{\set}[1]{\left\{ {#1} \right\}}
\newcommand{\prt}[1]{\left( {#1} \right)}

\newcommand{\scal}[1]{\left< {#1} \right>}

\newcommand{\eq}{\ =\ }
\newcommand{\setN}{{\mathbb N}}              

\newcommand{\setR}{{\mathbb R}}
\newcommand{\setC}{{\mathbb C}}

\newcommand{\implie}{\Rightarrow}

\newcommand{\Limplie}{\ \ \Longrightarrow\ \ }

\newcommand{\precc}{\prec\!\!\prec}
\newcommand{\A}{\mathcal{A}}
\newcommand{\B}{\mathcal{B}}

\renewcommand{\H}{\mathcal{H}}

\newcommand{\J}{\mathcal{J}}
\newcommand{\M}{\mathcal{M}}

\newcommand{\T}{\mathcal{T}}
\newcommand{\C}{\mathcal{C}}

\newcommand{\Sw}{\mathcal{S}}

\DeclareMathOperator{\Span}{span}

\DeclareMathOperator{\tr}{tr}

\begin{document}

\title{\vspace{-1.1cm}{\bf An algebraic formulation of causality for noncommutative geometry}\vspace{0.5cm}}

\author{Nicolas Franco$^a$ \and Micha{\l} Eckstein$^{b,a}$\vspace{0.3cm}}

\date{{\footnotesize $^a$ Copernicus Center for Interdisciplinary Studies{\footnote{\small supported by a grant from the John Templeton Foundation}},\\
 ul. S{\l}awkowska 17, 31-016 Krak\'ow, Poland\\[0.3cm]
 $^b$ Faculty of Mathematics and Computer Science, Jagellonian University,\\
 ul. {\L}ojasiewicza 6, 30-348 Krak\'ow, Poland\\[0.3cm]
 nicolas.franco@math.unamur.be \quad michal.eckstein@uj.edu.pl}}

\maketitle

\begin{abstract}
We propose an algebraic formulation of the notion of causality for spectral triples corresponding to globally hyperbolic manifolds with a well defined noncommutative generalization. The causality is given by a specific cone of Hermitian elements respecting an algebraic condition based on the Dirac operator and a fundamental symmetry. We prove that in the commutative case the usual notion of causality is recovered. We show that, when the dimension of the manifold is even, the result can be extended in order to have an algebraic constraint suitable for a Lorentzian distance formula.
\end{abstract}

\section{Introduction}

Lorentzian noncommutative geometry is a new extension of noncommutative geometry which tries to adapt Alain Connes' theory to Lorentzian spaces. The initial formulation of noncommutative geometry \cite{C94} is only suitable for Riemannian geometry, and its main application in physics, called almost commutative geometry or noncommutative standard model, offers a description of the standard model of particle physics coupled to gravity with Euclidean signature \cite{MC08}. The aim of Lorentzian noncommutative geometry is to find a way to apply the theory of spectral triples to Lorentzian geometry, with as a long term goal the construction of an almost commutative model including gravity with the correct signature. The adaptation of noncommutative geometry to pseudo-Riemannian and Lorentzian spaces has begun in the last decade \cite{Stro,Mor}. More recently, several physicists and mathematicians have become interested by the question.\\

\vspace*{0.5cm}

The goal of this paper is to show that it is possible to include within the definition of a spectral triple a notion of causality. The main difficulty in defining causal structures on noncommutative spaces is the lack of notion of points. In particular, any definition of causality that uses curves would not be suitable for a noncommutative generalization. The first step towards a noncommutative causal structure consists in rephrasing the classical one in a purely algebraic way. So the path we shall follow in this paper is the one of algebraisation of geometry with the final aim of a definition of causality that remain valid in the noncommutative regime.

The algebraic structure presented in this paper comes from the existence of causal functions on Lorentzian spaces. This is motivated by the use of causal functions in existing attempts to define a noncommutative Lorentzian distance \cite{Mor,F2,F3} and the definition of isocones by Fabien Besnard which give a generalization of partially ordered spaces \cite{Bes}. The extension of causal functions to noncommutative spaces is given by a fully algebraic condition \mbox{$\forall \phi \in \H, \scal{\phi,\J[D,f] \phi} \leq 0$}, involving the Dirac operator $D$ and the fundamental symmetry $\J$, which ensures that the constructed isocones correspond to a Lorentzian geometry. In this paper we  prove that, when the spectral triple is constructed from a globally hyperbolic manifold, the usual notion of causality is recovered. This gives us the opportunity to extend the definition of the causal structure to noncommutative spaces and to lift the veil on what could be causality in the noncommutative regime. Then we show that, for even dimensions, the algebraic condition for causal functions can be extended to unbounded functions in order to have a suitable constraint for the construction of a Lorentzian distance formula. We give the proof that, for a spectral triple constructed from the Minkowski spacetime, the usual notion of distance can be recovered.

The plan of this paper is the following: In Section \ref{lost} we explain the basic structures for Lorentzian noncommutative geometry which are Lorentzian spectral triples. In Section \ref{caus} we put forward the basic axioms of a causal cone that induces a partial order on the space of states. We give the proof that, in the commutative case, this partial order structure corresponds to the usual causal structure between the points of a globally hyperbolic manifold. In Section \ref{causnc} we extend the structure to noncommutative spacetimes and define the notion of causal and chronological futures for the space of states. Then in Section \ref{secdist} we show that our result can be applied in order to have a suitable algebraic constraint to define a Lorentzian distance formula for even dimensional manifolds. The Appendix contains some technical computations which are essential for the proofs of the main results of this paper.\\

\vspace*{0.5cm}

\section{Lorentzian spectral triple}\label{lost}

The mathematical structures used in this paper are Lorentzian spectral triples. However, the definition of a Lorentzian spectral triple is still a work in progress with different but similar proposals. Since our considerations could be applied within the different approaches, we will just highlight the axioms which are in common and significant for our result.

\begin{defn}\label{deflost}
A Lorentzian spectral triple (minimal set of axioms) is given by the data \mbox{$(\A,\widetilde\A,\H,D,\J)$} with:
\begin{itemize}
\item A Hilbert space $\H$.
\item A non-unital pre-$C^*$-algebra $\A$ with a faithful representation as bounded operators on $\H$.
\item A preferred unitization $\widetilde\A$ of $\A$ which is a pre-$C^*$-algebra with a faithful representation as bounded operators on $\H$ and such that $\A$ is an ideal of $\widetilde\A$.
\item An unbounded operator $D$ densely defined on $\H$ such that, $\forall a\in\widetilde\A$:
\begin{itemize}
\item $[D,a]$ extends to a bounded operator on $\H$,
\item $a(1 + \scal{D}^2)^{-\frac 12}$ is compact, with $\scal{D}^2 = \frac 12 (D D^* + D^* D)$.
\end{itemize}  
\item A bounded operator $\J$ on $\H$ such that:
\begin{itemize}
\item $\J^2=1$,
\item $\J^*=\J$,
\item $[\J,a]=0\quad\forall a\in\widetilde\A$,
\item $D^*=-\J D \J$,
\item $\J$ captures the Lorentzian signature of the metric.
\end{itemize}
\end{itemize}
\end{defn}

\begin{defn}
A Lorentzian spectral triple is even if there exists a $\mathbb Z_2$-grading $\gamma$ such that $\gamma^*=\gamma$, $\gamma^2=1$, $[\gamma,a] = 0\ \forall a\in\widetilde\A$, \mbox{$\gamma \J =- \J \gamma$} and \mbox{$\gamma D =- D \gamma $}.
\end{defn}

The operator $\J$ is a fundamental symmetry which turns the Hilbert space $\H$ into a Krein space \cite{Stro,Bog} with the indefinite inner product $(\cdot,\cdot)_{\J} = \scal{\cdot,\J \cdot}$, where $\scal{\cdot,\cdot}$ is the positive definite inner product on the Hilbert space. The condition $D^*=-\J D \J$ is equivalent to the fact that $iD$ is self-adjoint for the indefinite inner product $(\cdot,\cdot)_{\J}$ (Krein-self-adjoint).

The condition that $\J$ captures the Lorentzian signature of the metric must be clarified. Without such a condition, we have a general pseudo-Riemmanian spectral triple \cite{Stro} with no real control on the signature. The construction of $\J$ varies within the different approaches:
\begin{itemize}
\item In an approach of one of the authors \cite{F4,F5}, the fundamental symmetry is constructed as $\J=-[D,\T]$ from an unbounded self-adjoint operator $\T$ with domain \mbox{$\text{Dom}(\T)\subset \H$} such that \mbox{$\prt{1+ \T^2}^{-\frac{1}{2}}\in \widetilde\A$}. Such operator $\T$ represents a global time function for the Lorentzian spectral triple.
\item In another approach initiated by Mario Paschke \cite{Pas} (see also \cite{Verch11,Rennie12}), the fundamental symmetry is a one-form $\J=\sum_i J a^0_i J^{-1} a_i [D,b_i]$ where $a^0_i,a_i,b_i \in \widetilde\A$ and where $J$ is an antilinear isometry (reality structure) compatible with the structure of the spectral triple (see \cite{MC08}).
\end{itemize}

Explicit constructions of commutative Lorentzian spectral triples can be performed given a Lorentzian manifold, typically with conditions of global hyperbolicity, completeness and existence of a spin structure. We refer the reader to \cite{Beem, Lawson} for more details on the usual concepts of Lorentzian geometry, causality and spin geometry. By a complete Lorentzian manifold we understand the following: there exists a spacelike reflection - i.e. a linear map on the tangent bundle respecting $r^2=1$, $g(r\cdot,r\cdot) = g(\cdot,\cdot)$ such that $g^r(\cdot,\cdot)=g(\cdot,r\cdot)$ is a Riemannian metric - such that the manifold is complete under the metric $g^r$ (see \cite{Stro}).

We will work with a noncompact complete globally hyperbolic Lorentzian manifold $\M$ of dimension $n$ with a spin structure $S$. Our conventions will be the following ones: the signature of the Lorentzian metric $g$ is $(-,+,+,+,\dots)$ and the Clifford action ``$c$'' respects $c(u)c(v)+c(v)c(u)=2g^{-1}(u,v)$ for $u,v\in T^*\M$. For any local basis $(x^0,\dots,x^{n-1})$, we define the curved gamma matrices $\gamma^\mu=c(dx^\mu)$ respecting the anticommutation conditions  $\{\gamma^\mu,\gamma^\nu\}=2g^{\mu\nu}$ and such that $\gamma^0$ is anti-Hermitian and $\gamma^a$ are Hermitian for $a>0$.

\begin{defn}\label{commlost}
A commutative Lorentzian spectral triple on a complete globally hyperbolic Lorentzian manifold $\M$ is constructed in the following way:
\begin{itemize}
\item $\H = L^2(\M,S)$ is the Hilbert space of square integrable spinor sections over $\M$ .
\item $D = -i(\hat c \circ \nabla^S) = -i \gamma^{\mu} \nabla^S_\mu$ is the Dirac operator.
\item $\A \subset C^\infty_0(\M)$ and $\widetilde\A \subset C^\infty_b(\M)$ with pointwise multiplication are some appropriate sub-algebras\footnote{As an example, on Minkowski spacetime, one can take $\A$ to be the space of Schwartz functions and $\widetilde\A$ the space of smooth bounded functions with bounded derivatives.} of the algebra of smooth functions vanishing at infinity and the algebra of smooth bounded functions such that $\forall a\in\widetilde\A$, $[D,a]$ extends to a bounded operator on $\H$. The representation is given by pointwise multiplication on $\H$.
\item $\J=i\gamma^0$.\\
\end{itemize}

\vspace{0.5cm}

If $n$ is even, the $\mathbb Z_2$-grading is given by the chirality element: \[\gamma = (-i)^{\frac{n}{2} + 1} \gamma^0 \dots \gamma^{n-1}.\]
\end{defn}
The choice of the fundamental symmetry $\J = i\gamma^0$ guarantees that the operator $iD$ is Krein-self-adjoint and that a commutative Lorentzian spectral triple is a Lorentzian spectral triple in the sense of the Definition \ref{deflost} \cite{Stro}. We must remark that, if we want to respect the definition of a fundamental symmetry in \cite{Stro} (and our exact axioms in the Definition \ref{deflost}), we have to require that $\J^2 = -(\gamma^0)^2 = -g^{00} = 1$ (so require $\gamma^0$ to be the flat gamma matrix). This can be obtained for every globally hyperbolic Lorentzian manifold by a conformal transformation of the metric in order to get $g^{00} = -1$. The causal structure of $\M$ is completely independent of such conformal transformation. However, the results presented in this paper are still valid for a general curved operator $\J = i\gamma^0$ with $(\gamma^0)^2 = g^{00} < 0$, so we will keep the curved notation throughout the paper.

For the reader's convenience, we give an example of a noncommutative Lorentzian spectral triple. The construction uses the Moyal product on the Minkowski spacetime \cite{F5} (see \cite{Gayral} for the Riemannian version). The Lorentzian spectral triple is constructed in the following way:

\begin{itemize}
\item $\H = L^2(\setR^{1,n-1}) \otimes \setC^{2^{\lfloor{n/2}\rfloor}}$ is the Hilbert space of square integrable spinor sections over the Minkowski spacetime.
\item $\A = \prt{\Sw(\setR^{1,n-1}),\star}$ is the space of Schwartz functions endowed with the Moyal product:
\[
(f \star h) (x) = \frac{1}{(2\pi)^n}\int \int f(x-\frac 12 \Theta u)\ h(x+v)\ e^{-iu \cdot v} \;d^nu\;d^nv,
\]
where $\Theta$ is a real skewsymmetric $n \times n$ constant matrix.
\item $\widetilde\A = \prt{\B(\setR^{1,n-1}),\star}$ is the space of smooth bounded functions with bounded derivatives endowed with the Moyal product.
\item $D = -i \gamma^\mu\partial_\mu$ is the flat Dirac operator on $\H$.
\item $\J=-[D,x^0] = ic(dx^0)  = i\gamma^0$ where $x^0$ is the global time.
\end{itemize}

The action of $\A$, $\widetilde\A$ and $D$ on $\H$ is given by the Moyal left multiplication. $\A$ and $\widetilde\A$ are pre-C$^*$-algebras if equipped with the operator norm \cite{Gayral} and $\A$ is an ideal of $\widetilde\A$. Another example of  noncompact noncommutative Lorentzian spectral triple is given by the noncommutative Lorentzian cylinder \cite{Suij}.\\

\vspace{0.5cm}

\section{Lorentzian spectral triples and causality}\label{caus}

In this section, we present a construction based on the definition of a particular subset of Hermitian elements of the algebra, called causal cone, which induces a partial order relation on the space of states. We prove that for a commutative Lorentzian spectral triple one recovers the usual notion of causality.

Let $(\A,\widetilde\A,\H,D,\J)$ be a Lorentzian spectral triple with the minimal set of axioms as in the Definition \ref{deflost}. 

\begin{defn}\label{causcone}
A causal cone $\C$ is a subset of elements in $\widetilde\A$ such that:
\begin{enumerate}
\item[$(a)$] $\forall a\in\C,\quad a^*=a$ 
\item[$(b)$] $\forall a,b\in\C,\quad a+b\in\C$
\item[$(c)$] $\forall a\in\C,\forall \lambda\geq0,\quad \lambda a\in\C$ 
\item[$(d)$] $\forall x\in\setR,\quad x1\in\C$ 
\item[$(e)$] $\overline{\Span_{\setC}(\C)} = \overline{\widetilde\A}$
\item[$(f)$] $\forall a\in\C,\forall \phi \in \H, \quad\scal{\phi,\J[D,a] \phi} \leq 0$
\end{enumerate}
where the closure denotes the $C^*$-algebra completion.
\end{defn}

A causal cone can be seen as a dense subset of an isocone as defined in \cite{Bes}. It is proved that isocones are equivalent, in a category theoretical sense, to partial order sets with topological structure (ordered topological spaces) where the order is completely determined by the elements of the isocone. So an isocone is a good start in order to define causality, even if the category of ordered topological spaces contains more objects than the category of causal spaces. To be more precise, the axioms of isocones are the axioms $(a)$ to $(e)$ with the requirement that the set is closed, plus an additional axiom $(f')$ which gives a kind of lattice structure:
\begin{enumerate}
\item[$(f')$] $\forall a,b\in\C,\quad a\vee b\in\C \ \text{ and }\  a\wedge b\in\C$,\\
where $a\vee b = \frac{a+b}{2} + \frac{\abs{a-b}}{2}$ and $a\wedge b = \frac{a+b}{2} - \frac{\abs{a-b}}{2}$ with $\abs{a}=\sqrt{a^*a}$.
\end{enumerate}
However, this axiom is not necessarily respected for a dense subset of an isocone, and that is why we cannot require it in our definition. As we will show, the new axiom $(f)$ refines the definition of isocones in such a way that it corresponds to a causal structure on the underlying space. This construction is also somehow related to the definition of causal cones in \cite{Mor}.

If a causal cone exists, then it defines a partial order on the space of states. We recall that states on $\widetilde\A$ are positive linear functionals (automatically continuous) of norm one (more precisely, states are defined on the $C^*$-completion $\overline{\widetilde\A}$ and we consider the restriction of those states on $\widetilde\A$). The space of states is denoted by $S(\widetilde\A)$. It is a closed convex set (for the weak-$^*$ topology), and extremal points are called pure states, with the set of pure states denoted by $P(\widetilde\A)$. We can define a partial order on $S(\widetilde\A)$ (and \textit{a fortiori} on $P(\widetilde\A)$) in the following way:

\begin{defn}\label{defcc}
Let $\C$ be a causal cone. For every two  states $\chi,\xi\in S(\widetilde\A)$, we define
\[\chi \preceq \xi \quad \text{ iff }\quad \forall a\in\C,\ \chi(a) \leq \xi(a).\]
\end{defn}

This definition is coherent: since states are positive linear functionals, their values on Hermitian elements must be real, so the inequality is well defined. We might consider only maximal cones to avoid possible dependence on the choice of the cone.

\begin{prop}
The relation $\preceq$ defines a partial order relation on $S(\widetilde\A)$.
\end{prop}

\begin{proof}
The relation $\preceq$ is trivially reflexive and transitive. To check the antisymmetry, suppose that $\chi$ and $\xi$ are such that $\chi \preceq \xi$ and $\xi \preceq \chi$. In this case, we have that $\forall a\in\C,\ \chi(a) = \xi(a)$. Since $\overline{\Span_{\setC}(\C)} = \overline{\widetilde\A}$, by linearity and continuity of the states we have that $\forall a\in\widetilde\A,\ \chi(a) = \xi(a)$, so $\chi = \xi$.
\end{proof} 

Let us suppose now that the Lorentzian spectral triple is commutative and constructed as in the Definition \ref{commlost}. In this case the space of pure states corresponds to the space of characters, i.e.~the set of all non-zero *-homomorphisms. By the Gel'fand--Naimark theorem, the set of characters (spectrum) $\Delta(\A) = P(\A)$ can be identified with the manifold by $\forall p \in \M$, $p \leadsto \chi \in \Delta(\A)$ such that $ \forall f\in\A, \chi(f)=f(p)$. Pure states on $\A$ can easily be extended to pure states on $\widetilde\A$ using the Hahn--Banach theorem \cite{BraRo}. However, the space $P(\widetilde\A)$ contains too many states. In particular, it contains states whose kernel contains the sub-algebra $\A$. Such states correspond to a compactification of the manifold $\M$ (Stone--\v Cech compactification \cite{Wegge}), and they should be removed in order to recover the usual causality relation on $\M$.

The following is the main result of the paper. It ensures that the relation $\preceq$ on the space of states, as set in the Definition \ref{defcc}, is an algebraization of the usual notion of causality.\\

\begin{thm}\label{mainthm}
Let $(\A,\widetilde\A,\H,D,\J)$ be a commutative Lorentzian spectral triple constructed from a globally hyperbolic Lorentzian manifold $\M$ as in the Definition \ref{commlost}, and let us define the following subset:
\[
\M(\widetilde\A) = \set{\chi\in P(\widetilde\A)  : \A  \not\subset  \ker \chi} \subset S(\widetilde\A).
\]
 Then, 
 \[
\M(\widetilde\A) \cong \Delta(\A) \cong \M,
\]
and the partial order relation $\preceq$ on $S(\widetilde\A)$ restricted to $\M(\widetilde\A)$ corresponds to the usual causal relation on $\M$. \\
\end{thm}

We will devote the rest of this section to the proof of the Theorem \ref{mainthm}.\\

First, we need to show that the subset $\M(\widetilde\A)$ corresponds to the spectrum $\Delta(\A)$. From the definition of $\M(\widetilde\A)$, every character of $\M(\widetilde\A)$ is still a non-zero *-homomorphism if restricted to $\A$, so it is sufficient to show that every character on $\A$ extends uniquely to $\widetilde\A$ to get the bijection.

\begin{prop}
Let $\chi\in\Delta(\A)$ be a character. Then $\chi$ has a unique extension on $\widetilde\A$.
\end{prop}

\begin{proof}
The extension of $\chi$ can be chosen in such a way that the character (pure state) property still holds  on $\widetilde\A$ \cite[Prop. 2.3.24]{BraRo}. Since $\A \not\subset \ker \chi$, there exists $a\in\A$ such that $\chi(a) \neq 0$. Then for every $b\in\widetilde\A$, we have:
\[
\chi(ab) = \chi(a) \chi(b)  \Limplie \chi(b)= \frac{\chi(ab)}{\chi(a)}.
\]
Since $\A$ is an ideal of $\widetilde\A$, $ab\in\A$, the values of $\chi$ are uniquely determined by the values of its restriction to $\A$.
\end{proof}

The cornerstone of the proof of Theorem \ref{mainthm} is the set of causal functions. Let us recall the following definition \cite{Mor}:

\begin{defn}
A causal function on a Lorentzian manifold is a real-valued function which is non-decreasing along every future directed causal curve.
\end{defn}

\begin{prop}\label{cfdo}
Let $\M$ be a globally hyperbolic Lorentzian manifold, then the set of smooth bounded causal functions $\C(\M)\subset  C^\infty_b(\M)$ completely determines the causal structure on $\M$ by
\[\forall p,q\in\M,\qquad p \preceq q \quad \text{ iff }\quad \forall f\in\C(\M),\ f(p) \leq f(q).\]
Moreover, $\C(\M)$ respects the axioms $(a)$ to $(e)$ of a causal cone for some suitable unitization $\widetilde\A\subset  C^\infty_b(\M)$.
\end{prop}

\begin{proof}
The proof of this proposition uses the results of \cite{Bes}.

If $p \preceq q$, there exists at least one causal curve from $p$ to $q$. From the definition of a causal function $f\in\C(\M)$, it is obvious that $f(p) \leq f(q)$. Moreover, if $p\neq q$, there exists a smooth global time function $T$ (function increasing along every future directed causal curve) such that $T(p) < T(q)$, which implies that there exists $f=\tan^{-1}\circ T\in\C(\M)$ such that $f(p) < f(q)$. It remains to show that if $p \npreceq q$ and $q \npreceq p$, there exists at least one causal function $f\in\C(\M)$ such that $f(p) > f(q)$ (the cases where $p \npreceq q$ and $q \preceq p$ are already covered by symmetry). Let us take $q'$ in the future of $q$ such that $p \npreceq q'$ and $q' \npreceq p$ still hold. From the results of A.N. Bernal and M. S\'anchez on the smooth splitting theorem \cite{BS04,BS06}, we can extract a smooth global time function $T$ such that $T(p)=T(q')>T(q)$. By taking $f=\tan^{-1}\circ T\in\C(\M)$, we have $f(p) > f(q)$. 

Axioms $(a)$, $(b)$, $(c)$ and $(d)$ of the Definition \ref{causcone} are trivially satisfied for $\C(\M)$. To show that $\overline{\Span_{\setC}(\C(\M))} = \overline{\widetilde\A}$ for some unitization $\widetilde\A$, let us consider the space $V=\Span_\setR(\C(\M))$. Every element $h\in V$ may be written \mbox{$h=f-g$} where $f,g\in\C(\M)$, and since those functions are bounded they can be written as $f=f_+-\lambda 1$ and $g=g_+-\mu 1$, where $f_+$ and $g_+$ are some positive causal functions and $\lambda,\mu \geq 0$ are some non-negative constants. So every element $h\in V$ can be rewritten as the difference of two positive causal functions $h=(f_+ +\mu 1)-(g_+ + \lambda 1)$, and since the product of two positive elements in $\C(\M)$ remains in $\C(\M)$ (the product of two positive non-decreasing functions is still a non-decreasing function), this leads to the fact the $V$ is stable by products, hence $V$ is an algebra. Now, there exists a compactification $X$ of $\M$ such that $V$ can be seen as a sub-algebra of $C(X,\setR)$ which separates the points of $X$ \cite{Bes,Nachbin}. By applying the Stone-Weierstrass theorem, $V$ is dense in $C(X,\setR)$, so $\Span_{\setC}(\C(\M)) = V + iV$ is dense in the closed unital algebra $\overline{\widetilde\A} = C(X) \subset C(\beta\M) \cong C_b(\M)$ using the isomorphism provided by the Stone-\~Cech compactification $\beta\M$.
\end{proof}

We can notice that the Proposition $\ref{cfdo}$ is still valid for smooth bounded functions with bounded derivatives (or for other similar suitable restrictions), since causal relations are local and do not depend of the behaviour of the functions at infinity (i.e.~one can always make a smooth modification of a causal function such that the function remains unchanged on some compact set and the derivatives become bounded at infinity). 

The proof of Theorem \ref{mainthm} relies on the fact that the causal functions are exactly the functions respecting the axiom $(f)$ of a causal cone.\\

\begin{thm}\label{causalfthm}
Let $(\A,\widetilde\A,\H,D,\J)$ be a commutative Lorentzian spectral triple as in the Definition \ref{commlost}, then $f\in\widetilde\A$ is causal if and only if 
\[
\forall \phi \in \H, \quad\scal{\phi,\J[D,f] \phi} \leq 0,
\]
where $\scal{\cdot,\cdot}$ is the inner product on $\H$.\\
\end{thm}

\begin{proof}
A real-valued function $f$ is causal if it is non-decreasing along every future-directed causal curve. Since $f$ is differentiable, it is equivalent to require that $g(\nabla f, \nabla f) = g^{-1}(df,df) \leq 0$ with past-directed gradient everywhere. For a globally hyperbolic manifold $\M$, the time-orientation can be given by a global smooth time function $\T$ \cite{BS04,BS06} which allows us to set globally $x_0=\T$, and the gradient of $f$ is past directed if and only if $g(\nabla f,\nabla T) \leq 0$. It follows that $f$ is causal if and only if it satisfies the two following conditions:
\[
g(\nabla f,\nabla T) = g^{\mu\nu} f_{,\mu} T_\nu = g^{\mu0} f_{,\mu} \leq 0,
\]
\[
g(\nabla f, \nabla f) = g^{\mu\nu} f_{,\mu} f_{,\nu} = f_{,\mu}f^{,\mu} \leq 0,
\]
where $df=f_{,\mu} dx^\mu$.

Let us define $\alpha = f^{,0} = g^{0\mu}f_{,\mu}$ and $\beta = g^{00} (f_{,\mu}f^{,\mu}) = g^{00} g^{\mu\nu} f_{,\mu} f_{,\nu}$. Then  $f$ is causal if and only if
\[
\alpha = g^{0\mu}f_{,\mu} \leq 0  \quad\text{and}\quad \beta = g^{00} g^{\mu\nu} f_{,\mu} f_{,\nu} \geq 0,
\]
since we work with a metric such that $g^{00} < 0$.

From the well known property $[D,f]=-i\,c(df)$ \cite{Var}, we find that
\[\J[D,f] = i c(dx^0) \,(-i)\,c(df) = i\gamma^0  \,(-i)\,(\gamma^\mu f_{,\mu}) = \gamma^0\gamma^\mu f_{,\mu}. \]

Now, since $f_{,\mu}$ are continuous, the condition that 
\begin{equation}\label{inequint}
\forall \phi \in \H, \ \scal{\phi,\J[D,f] \phi}  = \int_{\M} \phi^* \J[D,f] \phi \;d\mu = \int_{\M} \phi^* (\gamma^0\gamma^\mu f_{,\mu}) \phi \;d\mu \leq 0
\end{equation}
is equivalent to require that
\begin{equation}\label{inequpoint}
\forall \phi \in \H, \quad \phi^* (\gamma^0\gamma^\mu f_{,\mu}) \phi  \leq 0
\end{equation}
at every point of the manifold. Indeed, \eqref{inequpoint} implies \eqref{inequint}, and if \eqref{inequpoint} is false at some point, then it must be false on some neighbourhood for some spinor vanishing outside this neighbourhood, and \eqref{inequint} cannot be true for this spinor (hence it cannot be true for all spinors in $\H$). 

In order to prove the theorem, it remains to show that, on every point of $\M$, the matrix $\J[D,f]=\gamma^0 \gamma^{\mu} f_{,\mu}$ is negative semi-definite in $\setC^{2^{\lfloor{n/2}\rfloor}}$ if and only if $\alpha\leq0$ and $\beta\geq0$.

First, let us assume that the dimension of the manifold $\M$ is $n=2$ or $3$. In this case, $\J[D,f]=\gamma^0 \gamma^{\mu} f_{,\mu}$ is a $2\times 2$ matrix which is negative semi-definite if and only if its trace is non-positive and its determinant is non-negative. The determinant is given by
\[
\det(\J[D,f]) = \exp(\tr(\log(\J[D,f]))) = \frac{1}{2} ( a_1^2 - a_2),
\]
where $a_k=\tr\left[(\J[D,f])^k\right]$. The values of $a_k$ are calculated in the Appendix (Lemma \ref{tmtr}) and give $a_1=2\alpha$ and $a_2=4\alpha^2 - 2\beta$. Hence $\J[D,f]$ is negative semi-definite if and only if
\[
\tr(\J[D,f])=2\alpha \leq 0 \quad\text{and}\quad\det(\J[D,f])= \frac{1}{2} ( a_1^2 - a_2) = \beta \geq 0.
\] 

Let us assume now that $n\geq 4$. The eigenvalues of $\J[D,f]=\gamma^0 \gamma^{\mu} f_{,\mu}$ are the roots of the characteristic polynomial
\[
\det(\lambda 1-\J[D,f])  =  \lambda^d + c_1 \lambda^{d-1} + \ldots + c_{d-1} \lambda + c_d,\quad d=2^{\lfloor{n/2}\rfloor} \geq 4,
\]
where the coefficients are given by the following formula:
\begin{equation}\label{coefck}
c_k = (-1)^k \sum_{i = 0}^{\lfloor \frac{k}{2} \rfloor}  \binom{\frac d2}{k-i}  \binom{k-i}{i}  (2\alpha)^{k-2i} \beta^i.
\end{equation}
We adopt the standard convention $\binom{a}{b} = 0$ for $b<0$ or $b > a$. The complete calculation of those coefficients requires some combinatorial computations and is given in the Appendix (Lemma \ref{coefappendix}).

From the Vieta's formulas, we know that the matrix $\J[D,f]$ is negative semi-definite if and only if $c_k \geq 0$ for every $k\in\set{1,2,\dots,d}$.

\begin{itemize}
\item {\it (necessary condition)} Let us assume that $\alpha\leq0$ and $\beta\geq0$. We have that $(-1)^k (2\alpha)^{k-2i} \beta^i$ is non-negative, and since the binomial coefficients in \eqref{coefck} are non-negative, we have $c_k \geq 0$ for $k\in\set{1,2,\dots,d}$.
\item {\it (sufficient condition)} Let us assume that $c_k \geq 0$ for $k\in\set{1,2,\dots,d}$. For $k=1$, we get $c_1 = - d \alpha \geq 0 \implie \alpha \leq 0$.

If $\alpha = 0$, then $c_2 = \binom{\frac d2}{1}  \binom{1}{1} \beta = \frac d2 \beta$ since the only term remaining in the sum is the term with index $i=1$. So $c_2 \geq 0$ implies $\beta \geq 0$.

If $\alpha < 0$, then we can consider $c_{k}$ with $k=d-1$. Since the binomial coefficient $\binom{\frac d2}{k-i}$ is not null only if $k-i = d-1-i\leq \frac d2$, we have $i\geq \frac d2-1$, and since the maximal value of $i$ is ${\lfloor \frac{k}{2} \rfloor} = {\lfloor \frac{d-1}{2} \rfloor} = \frac d2-1$, the only remaining term is the one with the index $i=\frac d2-1$, so
\[
c_{d-1} = -  \binom{\frac d2}{\frac d2-1} \binom{\frac d2-1}{\frac d2-1} 2\alpha \beta^{\frac d2-1} = - d \alpha \beta^{\frac d2-1} \geq 0.
\] 
Since $\alpha <0$ and $\frac d2-1$ is odd, we get $\beta \geq 0$. 
\end{itemize}
\end{proof}

\vspace{0.5cm}

With this theorem, the proof of Theorem \ref{mainthm} is now complete. We can add the following remarks:
\begin{itemize}
\item The condition $\forall \phi \in \H, \scal{\phi,\J[D,f] \phi} \leq 0$ could be replaced by the condition $\forall \phi \in \H, \scal{\phi,\J[D,f] \phi} \geq 0$ if one chooses the other signature convention $(+,-,-,-,\dots)$. This is somehow similar to the condition that a smooth curve $\gamma : \setR \rightarrow \M$ is causal if and only if $\forall t\in\setR, g(\gamma'(t),\gamma'(t)) \leq 0$ for the first convention, or if and only if $\forall t\in\setR, g(\gamma'(t),\gamma'(t)) \geq 0$ for the other convention. The sign of the inequality if also reversed if the opposite time-orientation is chosen.
\item A causal cone is automatically empty if one tries to make such construction on a Lorentzian manifold that is not causal.
Indeed, such manifold contains closed timelike curves (in particular, this is the case for compact spacetimes \cite{Gallo,Tipler,Geroch}). Causal functions are constant on such curves and the set of causal functions cannot separate the points on the curve. As a consequence, the condition $\overline{\Span_{\setC}(\C)} = \overline{\widetilde\A}$ cannot be fulfilled. This condition is related to the condition of existence of a causal structure on the space.\\
\end{itemize}

\vspace{0.5cm}

\section{Causality in the noncommutative regime}\label{causnc}

In the previous section, we have succeeded in defining the causal structure of a globally hyperbolic Lorentzian manifold in fully algebraic terms. By the standard correspondence principle of noncommutative geometry \cite{C94}, we extend the Definitions \ref{causcone} and \ref{defcc} to general noncommutative spaces understood in terms of Lorentzian spectral triples. In this section, we make use of those definitions to introduce the notions of causal and chronological futures (pasts) for the space of states.

Let \mbox{$(\A,\widetilde\A,\H,D,\J)$} be a (possibly noncommutative) Lorentzian spectral triple and let us suppose that a causal cone $\C$ exists. We denote by $\preceq$ the partial order relation on $S(\widetilde\A)$ induced by $\C$.

We propose the following definitions of causal futures (pasts) of states:
\begin{defn} \label{Jdef}
For every $\chi\in S(\widetilde\A)$, we define
\[J^+_S(\chi)=\set{\xi\in S(\widetilde\A) : \chi \preceq \xi}  \quad\text{and}\quad J^-_S(\chi)=\set{\xi\in S(\widetilde\A) : \xi \preceq \chi}, \]
for every $\chi\in P(\widetilde\A)$, we define
\[J^+_P(\chi)=\set{\xi\in P(\widetilde\A) : \chi \preceq \xi}   \quad\text{and}\quad  J^-_P(\chi)=\set{\xi\in P(\widetilde\A) : \xi \preceq \chi} \]
and for every $\chi\in \M(\widetilde\A) = \set{\chi\in P(\widetilde\A) : \A \not\subset \ker \chi}$, we define
\[J^+_\M(\chi)=\set{\xi\in \M(\widetilde\A) : \chi \preceq \xi}   \quad\text{and}\quad  J^-_\M(\chi)=\set{\xi\in \M(\widetilde\A) : \xi \preceq \chi}. \]
\end{defn}

It is obvious that $J^+_P(\chi) = J^+_S(\chi) \cap P(\widetilde\A)$ and $J^+_\M(\chi) = J^+_S(\chi) \cap \M(\widetilde\A)$ for a suitable $\chi$.

\begin{remark} \label{Sremark}
If two states $\chi$ and $\xi$ are such that $\chi \preceq \xi$, then for every convex combination $\eta = (1-\lambda)\chi + \lambda \xi \in S(\widetilde\A)$, $\lambda\in[0,1]$, we have:
\[
\forall a\in\C,\quad \chi(a) \leq  \eta(a) \leq \xi(a)  \quad\implies\quad \chi \preceq \eta \preceq \xi, 
\]
which implies that two states are causaly related if they are convex combinations of two causaly related pure states (but this is not a necessarily condition).
\end{remark}

We also have the following result:

\begin{prop}
For every $\chi\in S(\widetilde\A)$, $J^+_S(\chi)$ and $J^-_S(\chi)$ are two closed convex sets in $S(\widetilde\A)$ (for the weak-$^*$ topology) such that $J^+_S(\chi) \cap J^-_S(\chi) = \set{\chi}$.
\end{prop}

\begin{proof}
Let us suppose that $\xi_0,\xi_1\in J^+_S(\chi)$ and that $\xi_\lambda = (1-\lambda) \xi_0 + \lambda \xi_1$, $\lambda\in[0,1]$ is a convex combination.  Then, $\forall a\in\C$,
\[ \xi_\lambda (a)= (1-\lambda) \xi_0(a) + \lambda \xi_1(a) \geq (1-\lambda) \chi(a) + \lambda \chi(a) = \chi(a),\]
so $\xi_\lambda\in J^+_S(\chi)$. Moreover, if $\xi_n\in J^+_S(\chi)$ and $\xi_n \stackrel{w^*}{\rightarrow} \xi$, then $\forall a\in\C$, $\xi_n(a) {\rightarrow} \xi(a) \geq \chi(a)$ so $\xi\in J^+_S(\chi)$.

$J^+_S(\chi) \cap J^-_S(\chi) = \set{\chi}$ comes from the antisymmetry property of the partial order relation $\preceq$.
\end{proof}

Now, we let us recall that, for a Lorentzian manifold, the chronological future $I^+(p)=\set{q:p\precc q}$, which is the set of all points reached by a future directed timelike curve starting at $p$, corresponds to the topological interior of the causal future $J^+(p)=\set{q:p\preceq q}$ \cite{Beem}. This observation justifies the following definition of a chronological future (past) of a state based on the weak-$^*$ topology.

\begin{defn}
For every $\chi$ respectively in $S(\widetilde\A)$, $P(\widetilde\A)$ and $\M(\widetilde\A)$, we define the following chronological futures (pasts):
\begin{itemize}
\item $I^\pm_S(\chi)= S(\widetilde\A) \setminus \overline{\prt{ S(\widetilde\A) \setminus J^\pm_S(\chi) }} $,
\item $I^\pm_P(\chi)= P(\widetilde\A) \setminus \overline{\prt{ P(\widetilde\A) \setminus J^\pm_P(\chi) }} $,
\item $I^\pm_\M(\chi)= \M(\widetilde\A) \setminus \overline{\prt{ \M(\widetilde\A) \setminus J^\pm_\M(\chi) }} $,
\end{itemize}
where the closure denotes the weak-$^*$ closure.
\end{defn}
\begin{defn}
For every two pure states $\chi,\xi\in S(\widetilde\A)$ (respectively two states $\chi,\xi$ in $P(\widetilde\A)$ or $\M(\widetilde\A))$, we define
\[\chi \precc \xi \quad \text{ iff }\quad \xi \in I^+_S(\chi)\quad (\text{resp. }\ \xi \in I^+_P(\chi) \text{ or } \xi \in I^+_\M(\chi)).\]
\end{defn}

The following proposition shows that this definition of the chronological future for noncommutative spaces is coherent with the usual definition in the commutative case.

\begin{prop}
If $(\A,\widetilde\A,\H,D,\J)$ is a commutative Lorentzian spectral triple as in the Definition \ref{commlost}, then the chronological relation $\precc$ on $\M(\widetilde\A)$ corresponds to the usual chronological relation.
\end{prop}

\begin{proof}
We consider that the points $p,q,p_n$ ($n\in\setN$) are associated to the states $\chi,\xi,\chi_n\in\M(\widetilde\A)$ by $f(p)=\chi(f)$, $f(q)=\xi(f)$ and $f(p_n)=\chi_n(f)$ for all $f\in\widetilde\A$. We already know that $p_n \not\in J^+(q)$ if and only if $\chi_n \not\in J^+_\M(\xi)$ by Theorem \ref{mainthm}.

Let us suppose that $p\in J^+ (q)\setminus I^+ (q)$, then $\exists p_n \not\in J^+(q)$ such that $p_n \rightarrow p$. For all $f\in\widetilde\A$, we have that $f(p_n) \rightarrow f(p)$ which implies $\chi_n(f) \rightarrow \chi(f)$, so $\chi_n \stackrel{w^*}{\rightarrow} \chi$ with $\chi_n \not\in J^+_\M(\xi)$, hence $\chi\in J^+_\M (\xi)\setminus I^+_\M (\xi)$.

Now, let us suppose that $\chi\in J^+_\M (\xi)\setminus I^+_\M (\xi)$, then $\exists \chi_n \not\in J^+_\M (\xi)$ such that $\chi_n \stackrel{w^*}{\rightarrow} \chi$, which means that $\forall f\in\widetilde\A$, $\chi_n(f) \rightarrow \chi(f)$ which implies $f(p_n) \rightarrow f(p)$. Since $\widetilde\A$ separates the points we have $p_n \stackrel{w}\rightarrow p$, which is equivalent to the strong convergence since the spacetime is finite dimensional, hence $p\in J^+ (q)\setminus I^+ (q)$.\\
\end{proof}

\section{An algebraic constraint for the Lorentzian distance formula}\label{secdist}

In this section, we show that, if the dimension of the manifold is even, the result given in Theorem \ref{causalfthm} can be extended in order to have a suitable constraint for the algebraization of the Lorentzian distance formula.

We recall that a Lorentzian distance, also called proper time, on a Lorentzian manifold $\M$ is a function $d : \M \times \M \rightarrow [0,+\infty) \cup \set{+\infty}$ defined by
$$
d(p,q) =
\left\{ 
\begin{array}{ll}
\quad \sup \left\{ l(\gamma) : 
\begin{array}{c}
\gamma \text{ future directed causal}\\
\text{piecewise smooth curve}\\
\text{with } \gamma(0)=p,\ \gamma(1)=q
\end{array}
\right\} 
&\text{if } p \preceq q\\ 
\quad 0 &\text{if } p \npreceq q
\end{array}
\right.
 $$
 where $l(\gamma) = \int \sqrt{ -g_{\gamma(t)}( \dot\gamma(t) ,  \dot\gamma(t) )}\ dt$ is the length of the curve.\\
 
If $\M$ is globally hyperbolic, this function is finite and continuous \cite{Beem}, and respects the following properties:
\begin{enumerate}
\item $d(p,p) = 0$,
\item $d(p,q) \geq 0$ for all $p,q\in\M$,
\item if $d(p,q) > 0$, then $d(q,p) = 0$,
\item if $d(p,q) > 0$ and $d(q,r) > 0$, then $d(p,r) \geq d(p,q) + d(q,r)$.\begin{flushright} (``wrong way'' triangle inequality)\end{flushright}
\end{enumerate}

We have the following result on the algebraization of the Lorentzian distance formula \cite{F3}:
\begin{thm}\label{gec}
Let us consider a oriented time-oriented Lorentzian manifold $\M$ and two points $p,q\in\M$. Then
\[
d(p,q) \leq \inf\set{ \max\set{0,f(q)-f(p)} \ :\  \begin{array}{c} f \in C^\infty(\M,\setR),\\
 \sup g( \nabla f, \nabla f ) \leq -1 , \\
 \nabla f \text{ is past-directed} \end{array}\ } \cdot
\]
Moreover, the equality can always be obtained by extending the set of smooth functions to the set of continuous causal functions.
\end{thm}

The following theorem extends our algebraic constraint in order to have a control on the growth rate of smooth causal functions.\\

\vspace{0.5cm}

\begin{thm}\label{distfthm}
Let $(\A,\widetilde\A,\H,D,\J)$ be an even commutative Lorentzian spectral triple as in the Definition \ref{commlost} and let us consider the *-algebra $C^\infty(\M,\setR)$ of smooth functions, then $f\in C^\infty(\M,\setR)$ respects the following constraints
\[
\sup g( \nabla f, \nabla f ) \leq -1  \quad\text{and}\quad \nabla f \text{ is past-directed}
\]
if and only if 
\[
\forall \phi \in \H, \quad\scal{\phi,\J([D,f]+ i\gamma) \phi } \leq 0,
\]
where $\scal{\cdot,\cdot}$ is the inner product on $\H$ and $\gamma$ is the chirality element.
\end{thm}

\begin{proof}
Let us set $\alpha =  g^{0\mu}f_{,\mu}$ and $\beta = g^{00} (g^{\mu\nu} f_{,\mu} f_{,\nu} + 1) = g^{00} (g( \nabla f, \nabla f ) +1) $. We can restart the proof of Theorem \ref{causalfthm} with the new matrix  $\J([D,f]+ i\gamma) = \gamma^0 \gamma^{\mu} f_{,\mu} - \gamma^0\gamma$, using Lemma \ref{coefdist} (with $a=1$) for the trace of the powers of the matrix and the coefficients of the characteristic polynomial, to obtain that $\forall \phi \in \H, \scal{\phi,\J([D,f]+ i\gamma) \phi } \leq 0$ if and only if $\alpha \leq 0$ and $\beta \geq 0$ on every point of $\M$, which is equivalent to have $\sup g( \nabla f, \nabla f ) +1 \leq 0$ with past-directed gradient of $f$.
\end{proof}

\vspace{0.2cm}

Now let us define a new distance function using this new algebraic constraint.

\vspace{0.2cm}

\begin{defn}\label{defnewdist}
Let $(\A,\widetilde\A,\H,D,\J)$ be an even commutative Lorentzian spectral triple as in Definition \ref{commlost}, for every two points $p,q\in\M$ we define:
\[
\widetilde d(p,q) := \inf_{f \in C^\infty(\M,\setR)}\set{ \max\set{0,f(q)-f(p)} :
\forall \phi \in \H, \scal{\phi,\J([D,f]+ i\gamma) \phi } \leq 0} \cdot
\]
\end{defn}

\vspace{0.2cm}

\begin{prop}
The function $\widetilde d(p,q)$ from Definition \ref{defnewdist} respects all the properties of a Lorentzian distance.
\end{prop}

\begin{proof}	
We have trivially that $\widetilde d(p,p) = 0$ and $\widetilde d(p,q) \geq 0$ for all $p,q\in\M$. If $\widetilde d(p,q) > 0 $, then $\widetilde d(q,p) = 0 $, since for every function $f$, $\max\set{0,f(q)-f(p)} > 0$ $\implie$ $\max\set{0,f(p)-f(q)} = 0$.

The wrong way triangle inequality is valid since, for all $p,q,r$ such that $\widetilde d(p,q)>0$ and $\widetilde d(q,r)>0$, we have:
\begin{align*}
\widetilde d(p,r) &= \inf\set{ \max\set{0,f(r)-f(p)} : \forall \phi \in \H, \scal{\phi,\J([D,f]+ i\gamma) \phi } \leq 0}\\
&= \inf\{  (f(r)-f(q)) + (f(q)-f(p))  : \forall \phi \in \H, \scal{\phi,\J([D,f]+ i\gamma) \phi } \leq 0 \}\\
&\quad(\text{since }f(r)-f(q) > 0 \text{ and } f(q)-f(p) > 0)\\
&\geq \inf\set{   f(r)-f(q) \ : \forall \phi \in \H, \scal{\phi,\J([D,f]+ i\gamma) \phi } \leq 0}\\
&\quad +\ \inf\set{  f(q)-f(p)  \ : \forall \phi \in \H, \scal{\phi,\J([D,f]+ i\gamma) \phi } \leq 0 }\\
&= \widetilde d(p,q) + \widetilde d(q,r). 
\end{align*} 
\end{proof}

\vspace{0.5cm}

\begin{prop}
For every $p,q\in \M$, we have $d(p,q) \leq \widetilde d(p,q)$.
\end{prop}
\begin{proof}
This is a simple consequence of Theorem \ref{gec} and  Theorem \ref{distfthm}.
\end{proof}

\vspace{0.2cm}

One can notice that this proposition implies that $\widetilde d(p,q)>0$ only if $d(p,q)>0$, so only if $p\precc q$.

However, those results are not sufficient to prove that, in the general case, the two distances coincide. This is related to the fact that our algebraic constraint is only well defined for smooth functions while the only proof of the equality case existing at this time relies on non-smooth causal functions \cite{F3}. Nevertheless, the equality case can easily be obtained if the distance between two points can be determined by using only smooth functions, as it is the case for the Minkowski spacetime.

\vspace{0.3cm}

\begin{prop}
Let us consider the following commutative spectral triple on Minkowski spacetime with even dimension $n$:
\begin{itemize}
\item $\H = L^2(\setR^{1,n-1}) \otimes \setC^{2^{\lfloor{n/2}\rfloor}}$ is the Hilbert space of square integrable spinor sections over the Minkowski spacetime.
\item $\A = \Sw(\setR^{1,n-1})$ is the algebra of Schwartz functions with pointwise multiplication.
\item $\widetilde\A = \B(\setR^{1,n-1})$ is the algebra of smooth bounded functions with bounded derivatives with pointwise multiplication.
\item $D = -i \gamma^\mu\partial_\mu$ is the flat Dirac operator.
\item $\J=-[D,x^0] = ic(dx^0)  = i\gamma^0$ where $x^0$ is the time coordinate.
\item $\gamma = (-i)^{\frac{n}{2} + 1} \gamma^0 \dots \gamma^{n-1}.$
\end{itemize}
Then, the function
\[
\widetilde d(p,q) = \inf_{f \in C^\infty(\setR^{1,n-1},\setR)}  \set{ \max\set{0,f(q)-f(p)} :
\forall \phi \in \H, \scal{\phi,\J([D,f]+ i\gamma) \phi } \leq 0}
\]
corresponds to the usual Lorentzian distance.
\end{prop}

\vspace{0.2cm}

\begin{proof}
Let $p,q\in\setR^{1,n-1}$ be two points such that $p \precc q$. We want to show that there exists a function $f\in C^\infty(\setR^{1,n-1},\setR)$ respecting $\sup g( \nabla f, \nabla f ) \leq -1$ with past-directed gradient such that $f(q)-f(p) = d(p,q)$. Since the Minkowski spacetime is invariant by translation and that the conditions on the gradient of $f$ are also invariant by translation, we can assume for simplicity that $p$ is at the origin. Let us suppose that $q=(q_0,q_1,\dots,q_{n-1})$ in cartesian coordinates. Since the maximal geodesics of the Minkowski spacetime are given by straight lines, so we have $d(p,q) = q_0^2 - \sum_{i=1}^{n-1}q_i^2 = d > 0$. We can notice that, since $p \precc q$ with $p$ at the origin, the coordinate $q_0$ is always positive. The point $q$ can be described by using hyperbolic polar coordinates: 
\[
q=(d\cosh\theta,d_1\sinh\theta,d_2\sinh\theta,\dots,d_{n-1}\sinh\theta),
\]
where $d=\sqrt{\sum_{i=1}^{n-1}d_i^2}$ and $\tanh\theta=\frac{\sqrt{ \sum_{i=1}^{n-1}q_i^2}}{q_0}$.

Let us define the following affine function:
\[
f(x_0,\dots,x_{n-1})= \cosh\theta \,x_0 - \sinh\theta \sum_{i=1}^{n-1}\frac{d_i}{d} x_i. 
\]
Clearly, $f$ is smooth and 
\[
g( \nabla f, \nabla f ) = -\cosh^2\theta + \sinh^2\theta \frac{\sum_{i=1}^{n-1}d_i^2}{d^2} = -1,
\]
with past-directed gradient since, from $\cosh\theta > 0$, $f$ is increasing along the time coordinate. Moreover,
\[
f(q) - f(p) = f(q) = d \cosh^2\theta - \sinh^2\theta \frac{\sum_{i=1}^{n-1}d_i^2}{d} = d = d(p,q).
\]

So for every pair of points $p,q$ such that $p \precc q$, we have $d(p,q)=\widetilde d(p,q)$. If we assume that $p \preceq q$ and $p \not\!\!\precc q$ (which implies $d(p,q)=0$), then there exists a sequence of points $q_k \rightarrow q$ such that $q \precc q_k$, $0 < d(q,q_k) = \widetilde d(q,q_k) \rightarrow 0$ and $0 <  d(p,q_k) = \widetilde d(p,q_k) \rightarrow 0$. If we suppose $\widetilde d(p,q)>0$, then the wrong way triangle inequality gives $0 <  \widetilde d(p,q) \leq \widetilde d(p,q_k) - \widetilde d(q,q_k) \rightarrow 0$, so $\widetilde d(p,q)=0$.

It remains to show that the formula is null if $p \npreceq q$. We will do it by geometrical considerations. On the Minkowski spacetime, the intersection $J^+(p) \cap I^+(q)$ cannot be empty, and since $I^+(q) \not\subset J^+(p)$ there must exist a point $r\in I^+(q)$ in the boundary of $J^+(p)$, i.e.~such that $r\in J^+(p) \setminus I^+(p)$. Then we have $p \preceq r$ with $d(p,r)=\widetilde d(p,r)=0$ and for every $\epsilon>0$ there exists a function $f$ respecting $g( \nabla f, \nabla f )=-1$ with past-directed gradient and such that $f(r)-f(p) < \epsilon$. Since we have $d(q,r)>0$, we must have $f(q) < f(r) < f(p) + \epsilon$, so $f(q)-f(p) < \epsilon$ and $\widetilde d(p,q)=0$.
\end{proof}

\vspace{0.5cm}

We want to conclude this section by making a remark that the generalization of the Lorentzian distance formula given in Definition \ref{defnewdist} to noncommutative spacetimes is not straightforward. Indeed, this formula requires an extension to some *-algebra of unbounded elements, which is not included at this time in the definition of Lorentzian spectral triples.\\

\vspace{1cm}

\section{Outlook}

We have shown that a partial order structure within the space of states can be defined for a Lorentzian spectral triple. This partial order structure corresponds to the usual causal structure when the Lorentzian spectral triple is built from a globally hyperbolic manifold. This definition opens the door to various new possibility of introducing causal aspects in noncommutative geometry, with possible physical applications for example in noncommutative quantum field theory.

This definition raises also some new questions. For example, we have seen that, in the commutative case, we only need the subset $\M(\widetilde\A)$ of the space of all states to recover the causal structure. The passage from $P(\widetilde\A)$ to $\M(\widetilde\A)$ allows to recover the initial noncompact manifold instead of its compactification. On the other hand, the causality relation (Definition \ref{defcc}) is well defined for all states in $S(\widetilde\A)$. In the commutative case this observation is obsolete (see Remark \ref{Sremark}) and does not provide us any new information about the causal structure. However, in the noncommutative regime the three causal futures covered by Definition \ref{Jdef} might be of very different nature even in the case of simple deformations like the mentioned Moyal-Minkowski spacetimes. Since in the noncommutative regime the notion of points becomes meaningless, causality would also become a non-local quantity. The investigation of conceptual consequences of the proposed definition of causal structure for noncommutative Lorentzian spectral triples is the aim of our future work.

Also, with the algebraization of the global constraint for the Lorentzian distance formula in even dimensions, we have added a new step in the construction of such formula. The future work on this task will address the problem of obtaining a correspondence with the usual distance formula for general globally hyperbolic spacetimes. Moreover, it would require setting a correct background for Lorentzian spectral triples in order to extend this formula to noncommutative spacetimes.\\

\vspace{0.7cm}

\section*{Acknowledgement}

The authors would like to thank Andrzej Sitarz for his useful suggestions. This work was supported by a grant from the John Templeton Foundation.\\

\vspace{1cm}

\appendix
\section*{Appendix: The characteristic polynomial of $\J[D,f]$}
\setcounter{section}{1}

In this Appendix, we compute the necessary formulas for our main proofs, using some techniques of spin geometry and combinatorics.

We suppose that $f\in C^1(\M,\setR)$ where $\M$ is a $n$-dimensional globally hyperbolic Lorentzian spin manifold. Let us set $\alpha = f^{,0} = g^{0\mu}f_{,\mu}$, $\beta = g^{00} (f_{,\mu}f^{,\mu}) = g^{00} g^{\mu\nu} f_{,\mu} f_{,\nu}$ and define $d=2^{\lfloor{n/2}\rfloor}$.

\begin{lemma}\label{tmtr}
For $1 \leq k\leq d$, we have
\[
a_k = \tr\left[(\gamma^0 \gamma^{\mu} f_{,\mu})^k\right] =  \frac d2 \sum_{j=0}^{\lfloor \frac{k}{2} \rfloor} (-1)^j  \, \frac{k}{k-j} \, \binom{k-j}{j} (2\alpha)^{k-2j} \beta^j.
\]
\end{lemma}

\begin{proof}
Let us first set $a_0 = \tr(1) = d$ and compute $a_1$. We have
\[
a_1 = \tr (\gamma^0 \gamma^{\mu}) f_{,\mu} = g^{0\mu} \tr (1) f_{,\mu} = d  \alpha.
\]

Now take $k \geq 2$, then:
\begin{align*}
a_k & = \tr (\gamma^0 \gamma^{\mu} f_{,\mu})^k = \tr \gamma^0 \gamma^{\nu} \gamma^0 \gamma^{\rho} (\gamma^0 \gamma^{\mu} f_{,\mu})^{k-2} f_{,\nu} f_{,\rho} \\
& = 2 g^{0\nu} f_{,\nu} \tr \gamma^0 \gamma^{\rho} (\gamma^0 \gamma^{\mu} f_{,\mu})^{k-2} f_{,\rho} - \tr \gamma^0 \gamma^0 \gamma^{\nu} \gamma^{\rho} (\gamma^0 \gamma^{\mu} f_{,\mu})^{k-2} f_{,\nu} f_{,\rho}  \\
& = 2 \alpha \tr (\gamma^0 \gamma^{\mu} f_{,\mu})^{k-1} - g^{00} \tr \gamma^{\nu} \gamma^{\rho} (\gamma^0 \gamma^{\mu} f_{,\mu})^{k-2} f_{,\nu} f_{,\rho} \\
& = 2 \alpha a_{k-1} - g^{00} \tfrac{1}{2} \tr \{\gamma^{\nu},\gamma^{\rho}\} (\gamma^0 \gamma^{\mu} f_{,\mu})^{k-2} f_{,\nu} f_{,\rho} \\
& = 2 \alpha a_{k-1} - g^{00} g^{\nu\rho} f_{,\nu} f_{,\rho} \tr (\gamma^0 \gamma^{\mu} f_{,\mu})^{k-2} \\
& = 2 \alpha a_{k-1} - \beta a_{k-2}.
\end{align*}
We have extensively used the anticommutation rules $\{\gamma^\mu,\gamma^\nu\}=2g^{\mu\nu}$ and, in the line 4, we have taken advantage of the symmetry of exchanging the indices $\mu \leftrightarrow \rho$.

The resulting linear recurrence relation
\[
a_k = 2 \alpha a_{k-1} - \beta a_{k-2}, \qquad a_0 = d, \qquad a_1 = d  \alpha
\]
has the characteristic equation $X^2 - 2 \alpha X + \beta = 0$ whose roots are $\alpha \pm \sqrt{\alpha^2-\beta}$. This leads to the general formula

\begin{equation}\label{eqrec1}
a_k =  \tr\left[(\gamma^0 \gamma^{\mu} f_{,\mu})^k\right]  =\frac{d}{2} \left[ \left(\alpha - \sqrt{\alpha^2-\beta}\right)^k + \left(\alpha + \sqrt{\alpha^2-\beta}\right)^k \right]\cdot
\end{equation}

The reader might be suspicious about the $\sqrt{\alpha^2-\beta}$, which might, \textit{a priori}, be a complex number. However, the following computation shows that \eqref{eqrec1} is always real.

By use of the Binomial Theorem, we get
\[
a_k \ =\  d \sum_{i=0}^{\lfloor \frac{k}{2} \rfloor} \sum_{j = 0}^{i} (-1)^j \binom{k}{2 i} \binom{i}{j} \alpha^{k-2j} \beta^j 
\]
and if $k\geq 1$, by use of the following combinatorial relation
\[
\sum_{i=0}^{\lfloor \frac{k}{2} \rfloor} \binom{k}{2 i} \binom{i}{j} = \frac{k}{k-j} 2^{k-2j-1} \binom{k-j}{j} \quad\text{for}\quad j\leq {\lfloor \frac{k}{2} \rfloor},
\]
we get
\[
a_k \ =  \frac d2 \sum_{j=0}^{\lfloor \frac{k}{2} \rfloor} (-1)^j 2^{k-2j} \, \frac{k}{k-j} \, \binom{k-j}{j} \alpha^{k-2j} \beta^j.
\]
\end{proof}

\vspace{0.5cm}

\begin{lemma}\label{coefappendix}
The coefficients of the characteristic polynomial of the matrix $\gamma^0 \gamma^{\mu} f_{,\mu}$ read
\begin{equation}\label{calculcoef}
c_k = (-1)^k \sum_{i = 0}^{\lfloor \frac{k}{2} \rfloor} \binom{\frac d2}{k-i}  \binom{k-i}{i}  (2\alpha)^{k-2i} \beta^i  \quad \text{ for }\quad  0 \leq k \leq d.
\end{equation}

\end{lemma}

\begin{proof}

The Newton's identities allow us to express the coefficients of the characteristic polynomial of a $d$-dimensional matrix $A$ in terms of $a_k \eq \tr A^k$. Let us recall the following formula:
\[
\det(\lambda 1-A)  =  \lambda^n + c_1 \lambda^{d-1} + \ldots + c_{d-1} \lambda + c_d ,
\]
where the coefficients $c_k$ for $n\in\set{1,2,\dots,d}$ can be obtained recursively from the Newton's identities (see e.g.~\cite{Kalman}):
\[
a_k + a_{k-1} c_1 + \ldots + a_1 c_{k-1} + k c_k = 0.
\]
It is also logical to set $c_0=1$.

Since we already know the values of $a_k$ for the matrix $\gamma^0 \gamma^{\mu} f_{,\mu}$ from the Lemma \ref{tmtr}, it is sufficient to prove that, for every $1 \leq k \leq d$, we have
\begin{equation}\label{eqgoal}
\sum_{p=1}^k a_p c_{k-p} = -k c_k,
\end{equation}
where the coefficients $c_k$ are those given in \eqref{calculcoef}.

The LHS of \eqref{eqgoal} can be written as:
\begin{equation}\label{bl}
\sum_{p=1}^k a_p c_{k-p} = \sum_{l=0}^{\lfloor\frac{k}{2}\rfloor} \; b_l \; (2\alpha)^{k-2l} \beta^l.
\end{equation}

The coefficients $b_l$ in \eqref{bl} are obtained by taking the coefficients of $(2\alpha)^{k-2j} \beta^j$ in $a_p$ and the coefficients of $(2\alpha)^{k-2(l-j)} \beta^{(l-j)}$ in $c_{k-p}$ (for $j=0,\dots,\lfloor \frac{p}{2} \rfloor$). We get
\[
b_l = (-1)^{k} \frac d2 \sum_{p=1}^k \sum_{j=0}^{\lfloor \frac{p}{2} \rfloor}  (-1)^{j-p} \frac{p}{p-j} \binom{p-j}{j}  \binom{\frac d2}{k-p-l+j}  \binom{k-p-l+j}{l-j}. 
\]

Let us replace the couple of indices $(p,j)$ by a new one $(q,j)$ where $q = p-j$. Recall that $\binom{a}{b} = 0$ for $b<0$ or $b > a$, so the binomial coefficient $\binom{p-j}{j} $ yields non-null terms only if $2j\leq p$. So $q$ is non-negative and cannot be null since $p\neq 0$. The binomial coefficients $\binom{\frac d2}{k-p-l+j}$, $\binom{k-p-l+j}{l-j}$ imply that effectively we have $q=p-j \leq k-l$ and $j \leq l$, so $q$ goes from $1$ to $k-l$ and $j$ goes from $0$ to $l$, and we obtain
\begin{align}
 b_l &= (-1)^{k} \frac d2 \sum_{q=1}^{k-l} \sum_{j=0}^l  (-1)^{q} \frac{(q+j)}{q} \binom{q}{j}  \binom{\frac d2}{k-l-q}  \binom{k-l-q}{l-j} \nonumber\\
&= (-1)^{k} \frac d2 \sum_{q=1}^{k-l} (-1)^{q}  \binom{\frac d2}{k-l-q}  \left[ \sum_{j=0}^l   \frac{(q+j)}{q} \binom{q}{j}    \binom{k-l-q}{l-j} \right]. \label{eqqj}
\end{align}

\vspace{0.5cm}

The term inside the square brackets of \eqref{eqqj} is independent of $q$. Indeed,
\begin{align*}
&\sum_{j=0}^l   \frac{(q+j)}{q} \binom{q}{j}    \binom{k-l-q}{l-j} \\
&\qquad=   \sum_{j=0}^l   \binom{q}{j}    \binom{k-l-q}{l-j}   +   \sum_{j=0}^l   \frac{j}{q} \binom{q}{j}    \binom{k-l-q}{l-j} \\
&\qquad=   \sum_{j=0}^l   \binom{q}{j}    \binom{k-l-q}{l-j}   + \sum_{j-1=0}^{l-1}   \binom{q-1}{j-1}    \binom{k-l-q}{(l-1)-(j-1)} \\
&\qquad=  \binom{k-l}{l} + \binom{k-l-1}{l-1} \\
&\qquad=   \binom{k-l}{l} + \frac{l}{k(k-l)} \binom{k-l}{l}\\
&\qquad=  \frac{k}{(k-l)} \binom{k-l}{l},
\end{align*}
where we have applied two times the Vandermonde identity between the lines 3 and 4.

\vspace{0.3cm}

If we insert this result into \eqref{eqqj}, we get:

\begin{equation}\label{eqq}
  b_l = (-1)^{k} \frac d2 \frac{k}{k-l} \binom{k-l}{l} \sum_{q=1}^{k-l} (-1)^{q}  \binom{\frac d2}{k-l-q}. 
\end{equation}

\vspace{0.3cm}

The summand of \eqref{eqq} can be split into two parts via the relation
\[
\binom{\frac d2}{k-l-q} =  \binom{\frac d2-1}{k-l-q} +  \binom{\frac d2-1}{k-l-(q+1)}\cdot
\]

\vspace{0.5cm}

Since the sum over $q$ has an alternate sign $(-1)^{q}$, all terms cancel out except the two extrema:
\begin{itemize}
\item For $q=k-l$, the term $(-1)^{q}\binom{\frac d2-1}{k-l-q-1}$ remains but is null since $k-l-q-1=-1$.
\item For $q=1$, the term $ (-1)^{q}\binom{\frac d2-1}{k-l-q}$ remains and is equal to $ -\binom{\frac d2-1}{k-l-1}$.
\end{itemize}

Thus finally we get
\[
b_l = -(-1)^{k} \frac d2 \frac{k}{k-l} \binom{k-l}{l} \binom{\frac d2 -1}{k-l-1} = - k \left[(-1)^{k} \binom{\frac d2 }{k-l} \binom{k-l}{l}\right]
\]
where the term inside the square brackets clearly matches the coefficient of \eqref{calculcoef} for $i=l$.

It follows that the formula \eqref{eqgoal} is valid for the coefficients $c_k$ given in \eqref{calculcoef}, so they indeed correspond to the coefficients of the characteristic polynomial of $\gamma^0 \gamma^{\mu} f_{,\mu}$.
\end{proof}

\vspace{0.5cm}

The following is a generalization of the previous calculations to a matrix of the form $\gamma^0 (\gamma^{\mu} f_{,\mu} - a\gamma) = \gamma^0 \gamma^{\mu} f_{,\mu} - a \gamma^0\gamma$, where $\gamma$ is the chirality element.\\

\begin{lemma}\label{coefdist}
Let us suppose that $n$ is even and let $\gamma=(-i)^{\frac{n}{2} + 1} \gamma^0\cdots\gamma^{n-1}$. We set $\alpha = f^{,0} = g^{0\mu}f_{,\mu}$ and $\beta = g^{00} (f_{,\mu}f^{,\mu} + a^2) = g^{00} g^{\mu\nu} f_{,\mu} f_{,\nu} + g^{00} a^2$ for $a\in\setR$.

For $1 \leq k\leq d$, we have
\[
a_k = \tr\left[(\gamma^0 \gamma^{\mu} f_{,\mu} - a \gamma^0\gamma)^k\right] =  \frac d2 \sum_{j=0}^{\lfloor \frac{k}{2} \rfloor} (-1)^j  \, \frac{k}{k-j} \, \binom{k-j}{j} (2\alpha)^{k-2j} \beta^j
\]
and the coefficients of the characteristic polynomial of the matrix \mbox{$\gamma^0 \gamma^{\mu} f_{,\mu} - a \gamma^0\gamma$} read
\[
c_k = (-1)^k \sum_{i = 0}^{\lfloor \frac{k}{2} \rfloor} \binom{\frac d2}{k-i}  \binom{k-i}{i}  (2\alpha)^{k-2i} \beta^i  \quad \text{ for }\quad 0 \leq k \leq d.
\]

\end{lemma}

\begin{proof}
We abusively define $\gamma^n=\gamma=(-i)^{\frac{n}{2} + 1} \gamma^0\cdots\gamma^{n-1}$ and $f_{,n} = -a$. We use the following extended indices $\tilde\mu,\tilde\nu=0,1,\dots,n$ while the usual indices are $\mu,\nu=0,1,\dots,n-1$. We define $g^{nn}=1$, $g^{n\mu}=g^{\mu n}=0$ such that the anticommutation relation still holds $\{\gamma^{\tilde\mu},\gamma^{\tilde\nu}\} = 2 g^{\tilde\mu\tilde\nu}$.
Then we have,
\begin{align*}
& \gamma^0 \gamma^{\mu} f_{,\mu} - a \gamma^0\gamma = \gamma^0 \gamma^{\mu} f_{,\mu} + \gamma^0 \gamma^n f_{,n} = \gamma^0 \gamma^{\tilde\mu} f_{,\tilde\mu},\\
&g^{0\tilde\mu}f_{,\tilde\mu} = g^{0\mu}f_{,\mu} = \alpha,\\
& g^{00} g^{\tilde\mu\tilde\nu} f_{,\tilde\mu} f_{,\tilde\nu} = g^{00} g^{\mu\nu} f_{,\mu} f_{,\nu} + g^{nn} f_{,n} f_{,n} = g^{00} g^{\mu\nu} f_{,\mu} f_{,\nu} + g^{00} a^2 = \beta.
\end{align*}
The first result is obtained by applying the proof of Lemma \ref{tmtr}, using the extended indices, to the modified matrix $\gamma^0 \gamma^{\tilde\mu} f_{,\tilde\mu}$ with $\alpha=g^{0\tilde\mu}f_{,\tilde\mu}$ and $\beta= g^{00} g^{\tilde\mu\tilde\nu} f_{,\tilde\mu} f_{,\tilde\nu}$. The second result follows directly by applying Lemma \ref{coefappendix} to the modified matrix $\gamma^0 \gamma^{\tilde\mu} f_{,\tilde\mu}$.
\end{proof}


\begin{thebibliography}{99}

\bibitem{C94}
A. Connes, {\it Noncommutative Geometry}, Academic Press, \mbox{San Diego}, 1994.

\bibitem{MC08}
A. Connes and M. Marcolli, {\it Noncommutative Geometry, Quantum Fields and Motives}, American Mathematical Society, Colloquium Publications Vol. 55, Providence, 2008.

\bibitem{Stro}
A. Strohmaier, {\it On noncommutative and pseudo-Riemannian geometry}, J. Geom. Phys. 56 (2006) 175--195, math-ph/0110001.

\bibitem{Mor}
V. Moretti, {\it Aspects of noncommutative Lorentzian geometry for globally hyperbolic spacetimes}, Rev. Math. Phys. 15 (2003) 1171--1217, gr-qc/0203095.

\bibitem{F2}
N. Franco, {\it Towards a noncommutative version of Gravitation}, AIP Conference Proceedings 1241 (2010) 588--594, arXiv:1003.5407.

\bibitem{F3}
N. Franco, {\it Global Eikonal Condition for Lorentzian Distance Function in Noncommutative Geometry}, SIGMA 6 (2010) 064, arXiv:1003.5651.

\bibitem{Bes}
F. Besnard, {\it A noncommutative view on topology and order}, J. Geom. Phys. 59 7 (2009) 861--875, 	arXiv:0804.3551.

\bibitem{Bog}
J. Bognar, {\it Indefinite Inner Product Spaces}, Springer--Verlag, Berlin, 1974.

\bibitem{F4}
N. Franco, {\it Lorentzian approach to noncommutative geometry} (PhD thesis), Presses Universitaires de Namur, Namur, 2011, arXiv:1108.0592.

\bibitem{F5}
N. Franco, {\it Temporal Lorentzian Spectral Triple}, arXiv:1210.6575.

\bibitem{Pas}
M. Paschke and A. Sitarz, {\it Equivariant Lorentzian spectral triples}, math-ph/0611029.

\bibitem{Verch11}
R. Verch, {\it Quantum Dirac Field on Moyal-Minkowski Spacetime - Illustrating Quantum Field Theory over Lorentzian Spectral Geometry}, Acta Phys. Polonica B4 (Proc. Suppl.) (2011) 507--527, arXiv:1106.1138.

\bibitem{Rennie12}
K. van den Dungen, M. Paschke and A. Rennie, {\it Pseudo-Riemannian spectral triples and the harmonic oscillator}, arXiv:1207.2112.

\bibitem{Beem}
J. K. Beem, P. E. Ehrlich and K. L. Easley, {\it Global Lorentzian geometry}, 2nd ed., in Monographs and Textbooks in Pure and Applied Mathematics 202, Marcel Dekker, New York, 1996.

\bibitem{Lawson}
B. Lawson and L. Michelson, {\it Spin geometry}, Princeton University Press, Princeton, 1989.

\bibitem{Gayral}
V. Gayral, J. M. Gracia-Bondía, B. Iochum, T. Schücker and J.~C.~Varilly, {\it Moyal Planes are Spectral Triples}, Comm. Math. Phys. 246 (2004) 569--623, hep-th/0307241.

\bibitem{Suij}
W. D. van Suijlekom, {\it The noncommutative Lorentzian cylinder as an isospectral deformation}, J. Math. Phys. 45 (2004) 537--556, math-ph/0310009.

\bibitem{BraRo}
O. Bratteli and D. W. Robinson, {\it Operator algebras and quantum statistical mechanics. 1, C*- and W*-algebras symmetry groups decomposition of states}, 2nd ed., Springer, New York, 1987.

\bibitem{Wegge}
N. E. Wegge-Olsen, {\it K-theory and $C^*$-algebras: a friendly approach}, Oxford University Press, Oxford, 1993.

\bibitem{BS04}
A.N. Bernal and M. S\'anchez, {\it Smoothness of time functions and the metric splitting of globally hyperbolic spacetimes}, Comm. Math. Phys. 257 (2005) 43--50, gr-qc/0401112.

\bibitem{BS06}
A.N. Bernal and M. S\'anchez, {\it Further results on the smoothability of Cauchy hypersurfaces and Cauchy time functions}, Lett. Math. Phys. 77 (2006) 183--197, gr-qc/0512095.

\bibitem{Nachbin}
L. Nachbin, {\it Topology and order}, D. Van Nostrand Company, Princeton, 1965.

\bibitem{Var}
J. M. Gracia-Bondía, J. C. Várilly and H. Figueroa, {\it Elements of Noncommutative Geometry}, Birkhäuser, Boston, 2001.

\bibitem{Gallo}
G. Galloway, {\it Closed timelike geodesics}, Trans. Amer. Math. Soc. 285 (1984) 379--384.

\bibitem{Tipler}
F. J. Tipler, {\it Existence of a closed timelike geodesic in Lorentz spaces}, Proc. Amer. Math. Soc. 76 (1979) 145--147.

\bibitem{Geroch}
R. P. Geroch, {\it Topology in General Relativity}, J. Math. Phys. 8 (1967) 782--786.

\bibitem{Kalman}
D. Kalman, {\it A Matrix Proof of Newton's Identities}, Mathematics Magazine 73 (2000) 313--315.

\end{thebibliography}
\end{document}